\documentclass[11pt]{article}
\usepackage{amsgen,amsmath,amstext,amsbsy,amsopn,amssymb, bbm, algorithm, algorithmic, amsthm, enumerate, color, graphicx, tikz, natbib, pdflscape, multirow}
\usepackage[titletoc]{appendix}
\usepackage[colorlinks,citecolor=blue,urlcolor=blue]{hyperref}
\usepackage[left=2.54cm,top=2.54cm,right=2.54cm,bottom=2.54cm]{geometry}
\usetikzlibrary{intersections}

\newtheorem{Corollary}{Corollary}
\newtheorem{Proposition}{Proposition}
\newtheorem{Theorem}{Theorem}
\newtheorem{Lemma}{Lemma}

\newtheorem{Remark}{Remark}
\newtheorem{Assumption}{Assumption}
\newcounter{Facts}

\definecolor{ZurichRed}{rgb}{1, 0, 0} 
\definecolor{Gray}{gray}{0.85}
\newcommand{\ba}{{\mathbf{a}}}

\newcommand{\e}{{\mathbf{e}}}
\newcommand{\bu}{{\mathbf{u}}}
\newcommand{\bv}{{\mathbf{v}}}
\newcommand{\x}{{\mathbf{x}}}

\newcommand{\z}{{\mathbf{z}}}

\newcommand{\bA}{{\mathbf{A}}}

\newcommand{\R}{\mathbb{R}}

\newcommand{\Hsp}{{\mathcal{H}}}
\newcommand{\I}{{\mathbf{I}}}
\newcommand{\cJ}{{\mathcal{J}}}

\newcommand{\cN}{{\mathcal{N}}}

\newcommand{\cS}{{\mathcal{S}}}
\newcommand{\cT}{{\mathcal{T}}}

\newcommand{\X}{{\mathbf{X}}}

\newcommand{\bbeta}{{\boldsymbol{\beta}}}
\newcommand{\hbeta}{{\hat{\boldsymbol{\beta}}}}
\newcommand{\tbeta}{{\tilde{\boldsymbol{\beta}}}}
\newcommand{\bdelta}{{\boldsymbol{\delta}}}

\newcommand{\bgamma}{{\boldsymbol{\gamma}}}

\newcommand{\bmu}{{\boldsymbol{\mu}}}

\newcommand{\hsigma}{{\hat{\sigma}}}
\newcommand{\htau}{\hat{\tau}}
\newcommand{\btheta}{{\boldsymbol{\theta}}}
\newcommand{\bLambda}{{\boldsymbol{\Lambda}}}
\newcommand{\bOmega}{{\boldsymbol{\Omega}}}

\newcommand{\bSigma}{{\boldsymbol{\Sigma}}}
\newcommand{\hSigma}{{\hat{\boldsymbol{\Sigma}}}}
\newcommand{\tSigma}{{\tilde{\boldsymbol{\Sigma}}}}

\newcommand{\0}{{\mathbf{0}}}
\newcommand{\1}{{\mathbbm{1}}}

\newcommand{\half}{\frac{1}{2}}

\newcommand{\argmin}{\operatornamewithlimits{argmin}}
\newcommand{\argmax}{\operatornamewithlimits{argmax}}
\newcommand{\Corr}{{\rm Corr}}
\newcommand{\Cov}{{\rm Cov}}
\newcommand{\Var}{{\rm Var}}

\newcommand{\sign}{{\rm sign}}

\title{Optimal Estimation of Slope Vector in\\ High-dimensional Linear Transformation Model}
\author{Xin Lu Tan
\vspace{.3cm} \\
Department of Statistics \\
The Wharton School \\
University of Pennsylvania \\
Philadelphia, PA 19104}

\begin{document}

\maketitle

\begin{abstract}
In a linear transformation model, there exists an unknown monotone nonlinear transformation function such that the transformed response variable and the predictor variables satisfy a linear regression model. In this paper, we present CENet, a new method for estimating the slope vector and simultaneously performing variable selection in the high-dimensional sparse linear transformation model.  CENet is the solution to a convex optimization problem and can be computed efficiently from an algorithm with guaranteed convergence to the global optimum.  We show that under a pairwise elliptical distribution assumption on each predictor-transformed-response pair and some regularity conditions, CENet attains the same optimal rate of convergence as the best regression method in the high-dimensional sparse linear regression model.  To the best of our limited knowledge, this is the first such result in the literature.  We demonstrate the empirical performance of CENet on both simulated and real datasets.  We also discuss the connection of CENet with some nonlinear regression/multivariate methods proposed in the literature.
\end{abstract}


\section{Introduction}
\label{Sec:Intro}

Recently, there has been significant interest in theoretically appealing and algorithmically efficient regression methods that are suitable for analyzing high-dimensional datasets whose dimension~$p$ is comparable or much larger than the sample size~$n$.  Under the linear regression model, researchers have proposed various computationally efficient regularization methods  for simultaneously estimating the slope vector $\bbeta^*$ and selecting predictors, and substantial progress has been made on understanding the theoretical properties of these regression methods.  See, for example, \cite{Tibshirani1996, Chen1998, FanLi2001, Efron2004, Zou2005, Zou2006, Candes2007, Bickel2009, Buhlmann2011} and the references therein.  In particular, when $\bbeta^*$ is assumed to be sparse with number of non-zero coordinates $s = \|\bbeta^*\|_0 \ll p$, the lasso estimator and the Dantzig selector have been shown to both attain an optimal rate of convergence $\sqrt{s\log(p)/n}$ for estimating~$\bbeta^*$ \citep{Bickel2009}. 

When the relationship between the response variable and the predictors is nonlinear, however, the performance of regression methods based on linear models can be severely compromised.  In this paper, we study a family of sparse transformational regression models where a response variable $y\in\R$ is related to $p$ predictor variables $\x = (x_1, \ldots, x_p) \in \R^p$ through an equation\begin{equation}
\label{Eq:ModelTrans}
h^*(y) = \x^T\bbeta^* + \epsilon, \qquad\epsilon|\x\sim F(\epsilon).
\end{equation}
Here $h^*$ is an unknown, strictly increasing transformation function, $\bbeta^* = (\beta_1^*, \ldots, \beta_p^*)$ is a sparse slope vector of interest, and $\epsilon$ is a noise term independent of $\x$ with unknown distribution~$F$.  Model~\eqref{Eq:ModelTrans} generalizes the linear regression model by allowing an unknown monotone transformation on the response variable $y$.  It has been widely studied in the classical $n\gg p$ setting, but the problem of estimating $\bbeta^*$ in model~\eqref{Eq:ModelTrans} in the high-dimensional setting has received much less attention and a rate of convergence has not been derived.  \cite{HanWang2015} recently proposed a rank-based estimator for $\bbeta^*$ and established consistency but not rate of convergence of their estimator.  Furthermore, the algorithm for computing their estimator does not scale well with increasing dimension and sample size.  In this paper, we propose a computationally efficient method for estimating $\bbeta^*$ in the high-dimensional setting, given $n$ independent and identically distributed observations $\{(\x_i, y_i)\}_{i = 1}^n$ from \eqref{Eq:ModelTrans}.  Since $h^*$ is unknown, for identifiability of the model we assume that $\bbeta^*$ and $\epsilon$ satisfy $\bbeta^{*T}\bSigma_{xx}\bbeta^* = 1$ and $E(\epsilon) = 0$ throughout the paper, where $\bSigma_{xx} = \Cov(\x)$.  Interestingly, we showed that under certain regularity conditions, our estimator achieves a convergence rate of $\sqrt{s\log(p)/n}$ for estimating $\bbeta^*$.  To the best of our limited knowledge, this is the first result in the literature that shows that an estimator for model~\eqref{Eq:ModelTrans} achieves the same optimal rate of convergence as the best regression method in the sparse linear regression model, despite the unknown nonlinearity introduced by $h^*$.  

Model~\eqref{Eq:ModelTrans} is also known as the \emph{linear transformation model} \citep{Doksum1987, Dabrowska1988} and has been studied extensively in the literature in the low-dimensional setting.  It generates many popular econometric and statistical models under different assumptions on $h^*$ and $F$.  For example, when $h^*$ takes the form of a power function and $F$ is a normal distribution, model~\eqref{Eq:ModelTrans} reduces to the familiar Box-Cox transformation models \citep{BoxCox1964, BickelDoksum1981}.  On the other hand, model~\eqref{Eq:ModelTrans} specializes to the Cox's proportional hazard model \citep{Cox1972} and the proportional odds rate model \citep{Pettitt1982, Pettitt1984, Bennett1983} when $F$ denotes the extreme value distribution and the standard logistic distribution, respectively.  More general models for which $h^*$ is unknown and $F$ is specified up to a vector of finite-dimensional parameter were studied in \cite{Doksum1987, Dabrowska1988, Cheng1995, ChenJinYing2002}, among others.  In the case where there are no parametric specifications for both the transformation function and the noise distribution, \cite{Han1987} introduced the maximum rank correlation (MRC) estimator for $\bbeta^*$, treating $h^*$ as a nuisance parameter.  Given an estimate of $\bbeta^*$, \cite{Horowitz1996}, \cite{YeDuan1997}, \cite{Klein2002}, and \cite{Chen2002} further proposed methods for estimating $h^*$.  The MRC approach is based on the idea of maximizing a rank-based correlation between $y$ and $\x^T\bbeta$.  Rank estimation is known to be robust, and retains a relatively high level of efficiency.  However, due to the discontinuity of the rank-based objective function in the MRC approach, the computation of the MRC estimator is intractable especially in high dimensions.  In addition, such a discontinuity renders the theoretical analysis of the MRC estimator difficult.  \cite{LinPeng2013} proposed a computationally simpler smoothing approach as a remedy to the drawbacks of the MRC procedure, and established the $\sqrt{n}$-consistency and asymptotic normality of their estimator.  A more recent work by \cite{HanWang2015} offered a similar approach and extended the study of the estimation problem to the high-dimensional setting.  In fact, the model considered in \cite{HanWang2015} is a more general model than \eqref{Eq:ModelTrans}, and was first introduced in \cite{Han1987} under which the MRC estimator is also applicable.  Although \cite{HanWang2015} characterized the range of the problem dimensions for which their estimator is consistent, they did not establish the convergence rate of their estimator.  Moreover, their algorithm requires substantial running time for large $n$ and/or $p$ and only finds local optimums.  



Another line of work that is related to model~\eqref{Eq:ModelTrans} is based on the idea of sufficient dimension reduction.  In this framework, $y$ is assumed to be independent of $\x$ conditional on the projection of $\x$ onto some lower-dimensional subspaces.  The goal is then to recover the effective dimension reduction (e.d.r.) space, which is the minimal subspace for which the conditional independence property holds.  Model~\eqref{Eq:ModelTrans} can be viewed as a special case in this framework where the e.d.r.\! space of interest is one-dimensional.  Many methods have been proposed in the context of sufficient dimension reduction \citep{Cook1991, Li1992, Cook1998}, but the sliced inverse regression (SIR) method of \cite{DuanLi1991} is the first, yet the most widely studied method.  Although SIR can be applied to model~\eqref{Eq:ModelTrans}, the original implementation of SIR proposed by \cite{DuanLi1991} requires $p$ to be smaller than $n$.  While several proposals extending SIR to the high-dimensional setting have appeared in the literature \citep{LiYin2008, Yu2013, Lin2015}, the rate of convergence of SIR in high dimensions has not yet been derived.  In particular, \cite{LiYin2008} proposed SIR with both $\ell_1$- and $\ell_2$-regularization but did not study the theoretical properties of their method.  \cite{Yu2013} proposed to combine SIR with the Dantzig selector \citep{Candes2007} and established a non-asymptotic rate for the recovery of the e.d.r.\! space, but the sparsity~$s$ is not allowed to scale with $p$ and $n$ in their theoretical guarantees.  A more recent paper by \cite{Lin2015} proposed to apply a preliminary variables screening procedure before fitting SIR.  \cite{Lin2015} gave a consistency result for their method in the high-dimensional setting, but again, a rate of convergence has not been established. 

In this paper, we propose the Constrained Elastic Net (CENet), a new method that combines a robust rank correlation and an elastic net penalty for performing simultaneous estimation and variable selection in the linear transformation model~\eqref{Eq:ModelTrans} in the high-dimensional setting.  Our proposal is inspired by the correspondence between least squares regression and (one-dimensional) canonical correlation analysis (CCA), and is conceptually simple.  CENet is the solution of a convex optimization problem and hence can be computed efficiently.  We showed that under certain regularity conditions, the CENet estimator achieves the optimal convergence rate of $\sqrt{s\log(p)/n}$.  Our convergence rate is established from a random-$\x$ point of view, under a pairwise elliptical assumption on the distribution of each predictor and the transformed response.  We do not make any assumptions on $h^*$ besides it being strictly increasing.  In addition to robustness to nonlinearity in the response, our simulation results also indicate the robustness of CENet to heavy-tailed noise distribution.  Finally, we show that a connection exists between CENet and SIR, and also two other nonlinear multivariate techniques proposed in the literature: the alternating conditional expectation (ACE) method of \cite{Breiman1985} and the additive principal components (APC) method of \cite{DonnellBuja1994}.  Such an observation adds new insight to the understanding of behavior of SIR, ACE, and APC under model~\eqref{Eq:ModelTrans}.

\subsection{Related Models}
When $h^*$ is strictly increasing, model~\eqref{Eq:ModelTrans} belongs to the class of single-index models:
\begin{equation}
\label{ModelSIM}
y = f^*(\x^T\bbeta^*) + \tilde{\epsilon}, \qquad E(\tilde{\epsilon}|\x) = 0.
\end{equation}
This can be seen by noting the equivalence between model~\eqref{Eq:ModelTrans} and
\[y = g^*(\x^T\bbeta^* + \epsilon), \qquad\epsilon|\x\sim F(\epsilon),\]
where $g^* = (h^*)^{-1}$ is the inverse transformation, in which case $f^*(\x'\bbeta^*) = E(y|\x) = E[g^*(\x'\bbeta^* + \epsilon)|\x]$ and $\tilde{\epsilon} = y - E(y|\x)$.  In general, $f^*$ is different from $g^*$ due to the integration with respect to the noise distribution.

The single-index model is well-studied in low dimensions.  Some important proposals for estimating the slope vector $\bbeta^*$ include the semiparametric $M$-estimation approaches \citep{Ichimura1993, HardleHall1993, Delecroix2006} and the average derivative estimation approaches \citep{HardleStoker1989, Powell1989, Hristache2001}.  The $M$-estimators have been shown to possess nice theoretical properties such as consistency and asymptotic efficiency.  Nonetheless, they are rarely implemented in practice as they involve solving difficult non-convex optimization problem.  The average derivative estimators can be computed directly, but they involve local smoothing and yield worse performance as the dimension increases due to data sparseness (the so-called ``curse of dimensionality").  Another drawback for both the $M$-estimators and the average derivative estimators is that they usually require strong smoothness assumptions on $f^*$.  

In the high-dimensional setting, \cite{AlquierBiau2013} provided a PAC-Bayesian analysis for single index models in the case where $\bbeta^*$ is sparse.  \cite{Plan2014} and \cite{PlanVershynin2015} proposed metrically projected marginal regression estimators and generalized Lasso estimators, respectively, in the case where $\bbeta^*$ lies in some known low-dimensional subspace $K\subset\R^p$.  They showed that by exploiting the structural assumptions on  $\bbeta^*$, these estimators attain fast rates of convergence despite the unknown nonlinearity in observations.  However, the clean theoretical results of \cite{Plan2014} and \cite{PlanVershynin2015} are established under the assumption that $\x$ follows a multivariate normal distribution, and it is unclear if similar results hold at a greater generality.  \cite{Radchenko2015} proposed a promising approach for sparse single index models that does not require such a restrictive assumption on $\x$.  Nonetheless, the theoretical analysis requires boundedness on $\x$, smoothness on $f^*$, and subgaussianity of $\epsilon$, and the convergence rate of the proposed estimator is not as fast as ours.

\subsection{Organization}
The remainder of the paper is organized as follows.  In Section~\ref{Sec:Method}, we introduce our new methodology CENet, discuss the intuition behind it, and present an efficient algorithm for computing it.  We study the performance of CENet in Section~\ref{Sec:Main}.  In Section~\ref{Sec:RateofConv}, we establish the rate of convergence of CENet.  We conduct some simulation studies to evaluate the finite-sample performance of CENet in Section~\ref{Sec:Simulation}.  Section~\ref{Sec:Data} gives a real data application of CENet.  We conclude with Section~\ref{Sec:Discussion}, where we give a detailed description of the connection of CENet with SIR, ACE, and APC.  The proofs of the theoretical results in Section~\ref{Sec:RateofConv} are collected in Appendix~A.  Appendix~B contains the proofs of all other results in the paper, while Appendix~C contains the derivation details of our algorithm.  More extensive simulation results can be found in Appendix~D.

\paragraph{\textbf{Notation:}} For any event $E$, $\1(E)$ denotes its indicator function.  For a vector $\ba = (a_1, \ldots, a_p)^T \in\R^p$, we denote the $\ell_0$, $\ell_1$, $\ell_2$, and $\ell_\infty$ (pseudo)-norms of $\ba$ by $\|\ba\|_0 = \sum_{i=1}^p\1(a_i\neq 0)$, $\|\ba\|_1 = \sum_{i=1}^p |a_i|$, $\|\ba\| = (\sum_{i=1}^pa_i^2)^{1/2}$, and $\|\ba\|_\infty = \max_{1\leq i\leq p}|a_i|$, respectively.  For a matrix $\bA = (a_{ij}) \in\R^{p\times q}$, its spectral norm is given by $\|\bA\| = \sup_{\|x\|\leq 1}\|\bA\x\|$.  The notation $[k]$ stands for the index set $\{1, 2, \ldots, k\}$ for a positive integer $k$.  For any set $\cS$, $|\cS|$ denotes its cardinality and $\cS^c$ denotes its complement.  For any subsets $\cS\subseteq[p]$, $\cT\subseteq[q]$, we use $\bA_{\cS\cT}$ to represent the submatrix of $\bA$ with rows and columns indexed by $\cS$ and $\cT$.  For a symmetric matrix $\bA$, we write $\bA\succeq 0$ to indicate that $\bA$ is positive semidefinite.  For any $\bA\succeq 0$, $\bA^{1/2}$ denotes its principal square root that is positive semidefinite and satisfies $\bA^{1/2}\bA^{1/2} = \bA$.  The notation $\lambda_{\min}(\bA)$ and $\lambda_{\max}(\bA)$ stand for the smallest and the largest eigenvalue of $\bA$, respectively.  We write $\I$ for the identity matrix and $\0$ for the vector of all zeros (the respective dimension of which will be clear from context).  For any $a\in\R$, we set $\sign(a) = \1(a>0)-\1(a<0)$. Finally, $c, C, C'$, and variants thereof denote generic positive constants whose value may vary for each occurrence. 

\section{Methodology}
\label{Sec:Method}

In this section, we introduce \emph{CENet} for estimation of $\bbeta^*$ in the linear transformation model.  We are particularly interested in the setting where the dimension $p$ can be much larger than the sample size $n$.  In such a setting, it is common to assume that $\bbeta^*$ is sparse, i.e., $\bbeta^*$ contains many zero entries.

\subsection{The Optimization Problem}
Throughout the paper, we assume model~\eqref{Eq:ModelTrans} holds and we observe $n$ independent and identically distributed pairs $\{(\x_i, y_i): i = 1, \ldots, n\}$ generated from \eqref{Eq:ModelTrans}, where $\x_i = (x_{i1}, \ldots, x_{ip})\in\R^p$ and $y_i\in\R$.  We do not pose any parametric assumptions on the transformation function $h^*$, except that it being strictly monotonically increasing.  Let $\bSigma_{xx} = \Cov(\x)$ and $\bSigma_{xy} = \Cov(\x, h^*(y))$.  Since $h^*$ is unknown, for identifiability of the model we assume that any additive and multiplicative constants are absorbed into $h^*$, so that $\bbeta^{*T}\bSigma_{xx}\bbeta^* = 1$ and $E(\epsilon) = 0$.  Without loss of generality, we also assume that $E(\x) = \0$.  

In practice, when $h^*(y)$ is known, then $\bbeta^*$ can be found as the solution to  
\[
\min_{\bbeta\in\R^p} E[(h^*(y) - \x^T\bbeta)^2].
\]
Since $h^*$ is generally unknown, $\bbeta^*$ cannot be estimated through a direct least squares approach.  However, there is a direct correspondence between least squares regression and (one-dimensional) CCA:
\begin{Proposition}
\label{Prop:OurRep}
Suppose that model~\eqref{Eq:ModelTrans} holds.  Then $\bbeta^*$ is a solution to
\begin{equation}
\label{Eq:Corrstar}
\max_{\bbeta\in\R^p} \Corr(h^*(y), \x^T\bbeta).
\end{equation}
\end{Proposition}

In other words, among all linear combinations $\x^T\bbeta$ with $\Var(\x^T\bbeta) = 1$, $\x^T\bbeta^*$ has the highest correlation with $h^*(y)$.  Note that $\bbeta^*$ is only identifiable up to scale in \eqref{Eq:Corrstar}.  The following lemma gives an alternative formulation of \eqref{Eq:Corrstar}, where $\bbeta^*$ with $\bbeta^{*T}\bSigma_{xx}\bbeta^* = 1$ is the unique solution: 

\begin{Lemma}
\label{Lemma:PopRep}
Solving \eqref{Eq:Corrstar} is equivalent to solving the following CCA problem:
\begin{equation}
\label{Eq:CCAmax}
\max_{\bbeta\in\R^p} \bbeta^T\bSigma_{xy} \qquad\text{subject to}\qquad \bbeta^T\bSigma_{xx}\bbeta \leq 1.
\end{equation}
\end{Lemma}

To obtain an estimator $\hbeta$ for $\bbeta^*$, a simple idea is to replace $\bSigma_{xx}$ and $\bSigma_{xy}$ in \eqref{Eq:CCAmax} with their empirical estimates $\hSigma_{xx}$ and $\hSigma_{xy}$, and solve the resulting optimization problem.  However, such an approach is problematic in two respects:  (1) without knowledge of $h^*$, $\bSigma_{xy}$ is not estimable; (2) the sample covariance matrix $\hSigma_{xx} = \frac{1}{n}\sum_{i=1}^n\x_i\x_i^T$ is a natural estimator for $\bSigma_{xx}$ but is singular when $p\gg n$, so there might exist infinitely many solutions to such an empirical version of \eqref{Eq:CCAmax}.

To address the first issue, we note that for $j = 1, \ldots, p$,
\[(\bSigma_{xy})_j = \Cov(x_{1j}, h^*(y_1)) = \sigma_y\sigma_j\rho_j,\]
where $\sigma_y = [\Var(h^*(y_1))]^{1/2}$, $\sigma_j = [\Var(x_{1j})]^{1/2}$, and $\rho_j = \Corr(x_{1j}, h^*(y_1))$.  To estimate $\sigma_j$, we use $\hat{\sigma}_j = [(\hSigma_{xx})_{jj}]^{1/2}$, the sample standard deviation of $\{x_{1j}, \ldots, x_{nj}\}$.  To estimate $\rho_j$, we exploit the monotonic assumption on $h^*$ and consider the use of Kendall's tau, a classical nonparametric correlation estimator.  The Kendall's tau statistic is defined as
\begin{equation}
\label{tau}
\htau_j = \frac{2}{n(n-1)}\sum_{i<\ell}\sign(x_{ij}-x_{\ell j})\sign(y_i-y_\ell),
\end{equation}
Note that the Kendall's tau statistic is invariant under strictly monotone transformations of $x_{ij}$'s and $y_i$'s.  The population version of Kendall's tau is given by 
\begin{align*}
\tau_j &= E[\sign(x_{1j}-x_{2j})\sign(y_1-y_2)] \\
&= E[\sign(x_{1j}-x_{2j})\sign(h^*(y_1)-h^*(y_2))],
\end{align*}
where the second equality holds under the assumption that $h^*$ is strictly monotonically increasing.  

One can show that Kendall's tau $\tau_j$ is related to Pearson correlation $\rho_j$ through
\begin{equation}
\label{Eq:rho}
\rho_j = \sin\left(\frac{\pi}{2}\tau_j\right),
\end{equation}
when $(\x, \epsilon)$ follows a multivariate normal distribution (see, e.g., \cite{Cramer1946}, p. 290), or, more generally, a joint elliptical distribution \citep{Lindskog2003}.  Motivated by such a relationship, we propose to estimate $\rho_j$ by the transformed version of Kendall's tau:
\begin{equation}
\label{Eq:hrho}
\hat{\rho}_j = \sin\left(\frac{\pi}{2}\hat{\tau}_j\right).
\end{equation}
We then define as $\hSigma_{xy}$ our ``estimator" for $\bSigma_{xy}$, with $j^{th}$ entry given by
\begin{equation}
\label{Eq:Estxy}
(\hSigma_{xy})_j = \sigma_y\hat{\sigma}_j\sin\left(\frac{\pi}{2}\hat{\tau}_j\right).
\end{equation}

\begin{Remark}
Since $\bbeta^*$ is only identifiable up to scale, there is no need to estimate~$\sigma_y$.  Without loss of generality, we can set $\sigma_y$ in \eqref{Eq:Estxy} to 1. 
\end{Remark}

We wish to estimate $\bbeta^*$ in the case $p\gg n$, and also to induce sparsity on our estimator~$\hbeta$.  To address the second issue of singularity of $\hSigma_{xx}$, we propose to solve the following constrained optimization problem:
\begin{equation}
\label{Eq:Opt}
\min_{\bbeta\in\R^p} \left\{-\bbeta^T\hSigma_{xy} + \alpha_1\|\bbeta\|_1+ \alpha_2\|\bbeta\|^2\right\} \qquad\text{subject to}\qquad \bbeta^T\hSigma_{xx}\bbeta \leq 1.
\end{equation}
Here we turn the optimization problem \eqref{Eq:CCAmax} into a convex problem and we regularize the optimization criterion $-\bbeta^T\hSigma_{xy}$ by an $\ell_1+\ell_2$ elastic net penalty \citep{Zou2005}, where $\alpha_1\geq 0$ and $\alpha_2\geq0$ are tuning parameters.  We refer to the solution to \eqref{Eq:Opt} as \emph{CENet}, standing for Constrained Elastic Net.  

The formulation~\eqref{Eq:Opt} of CENet is inspired by the formulation of the CCA problem in \cite{Gao2014}.  Different from the CCA formulation in \cite{Gao2014} which only employs an $\ell_1$-penalty, in \eqref{Eq:Opt} an additional $\ell_2$-penalty is needed because of the use of the robust estimator $\hSigma_{xy}$.  In such a case, when $p>n$, $\hSigma_{xx}$ is singular and there might exist vectors belonging to the null space of $\hSigma_{xx}$ that are not orthogonal to $\hSigma_{xy}$, such that the criterion in \eqref{Eq:Opt} with $\alpha_2 = 0$ is unbounded below.  Indeed, the following lemma gives a sufficient condition for the existence and uniqueness of solution to problem~\eqref{Eq:Opt} across all values of $n$ and $p$.  The key is that $\alpha_2>0$ ensures that we have a strictly convex criterion.
\begin{Lemma}
\label{Lemma:Soln}
Suppose that $\alpha_1\geq 0$ and $\alpha_2>0$.  Then there exists a unique solution $\hbeta$ to problem~\eqref{Eq:Opt}.
\end{Lemma}

The following lemma, on the other hand, gives the range of $(\alpha_1, \alpha_2)$ for which $\hbeta = \0$.
\begin{Lemma}
\label{Lemma:CompSparse}
The solution to \eqref{Eq:Opt} is completely sparse if and only if $\alpha_1\geq\|\hSigma_{xy}\|_\infty$.
\end{Lemma}
Hence, when selecting $\alpha_1$, we need never consider a value larger than that in Lemma~\ref{Lemma:CompSparse}.

\subsection{ADMM Algorithm for CENet}
Although \eqref{Eq:Opt} is a convex problem, closed form solution is not available.  We consider using the Alternating Direction Method of Multipliers (ADMM) to solve \eqref{Eq:Opt} \citep{Boyd2011}.

We first rewrite the optimization problem as follows:
\begin{equation}
\label{Eq:ADMM}
\min_{\bbeta\in\R^p} f(\bbeta) + g(\btheta)\qquad\text{subject to}\qquad \hSigma_{xx}^{1/2}\bbeta - \btheta = \0,
\end{equation}
where
\[f(\bbeta) = -\bbeta^T\hSigma_{xy} + \alpha_1\|\bbeta\|_1 + \alpha_2\|\bbeta\|^2, \qquad g(\btheta) = \infty\1(\|\btheta\|>1).\]  
Then the augmented Lagrangian of \eqref{Eq:ADMM} takes the form
\[L_\eta(\bbeta, \btheta, \bgamma) = f(\bbeta) + g(\btheta) + \bgamma^T(\hSigma_{xx}^{1/2}\bbeta-\btheta) + \frac{\eta}{2}\|\hSigma_{xx}^{1/2}\bbeta-\btheta\|^2,\]
where $\eta>0$ is an ADMM parameter.  We present the ADMM algorithm for solving \eqref{Eq:Opt} in Algorithm~\ref{Algo}.  Derivations of the algorithm can be found in Appendix~C.

\begin{algorithm}[!h]
\caption{ADMM algorithm for CENet}
\begin{algorithmic}
\small
\STATE \noindent{{\bf Inputs:} Estimates $\hSigma_{xx}$ and $\hSigma_{xy}$, penalty parameters $\alpha_1$ and $\alpha_2$, ADMM parameter $\eta$, and tolerance level $\epsilon$.}
\vspace{-0.5em}  
\begin{enumerate}[(1)] \itemsep -0.25em
\item Initialize $k = 0$, $\btheta^0 = 0$, $\bgamma^0 = 0$.
\item Repeat:
\begin{enumerate}[(a)] \itemsep -0.75em
\vspace{-1em}
\item update for $\bbeta$ (solving a lasso problem):
\vspace{-0.5em}
\[\bbeta^{k+1} \leftarrow  \argmin_{\bbeta\in\R^p} \left\{\left\|\z - \left(\hSigma_{xx} + \frac{2\alpha_2}{\eta}\I\right)^{1/2}\bbeta\right\|^2 + \frac{2\alpha_1}{\eta}\|\bbeta\|_1\right\},\] 
where
\vspace{-0.5em}
\[\z = \left(\hSigma_{xx} + \frac{2\alpha_2}{\eta}\I\right)^{-1/2}\left(\hSigma_{xx}^{1/2}\btheta^k + \frac{1}{\eta}\hSigma_{xy} - \frac{1}{\eta}\hSigma_{xx}^{1/2}\bgamma^k\right);\]

\item update for $\btheta$ (projection onto the unit ball):
\vspace{-0.5em}
\[\btheta^{k+1} \leftarrow \frac{\hSigma_{xx}^{1/2}\bbeta^{k+1} + \frac{1}{\eta}\bgamma^k}{\max\left\{\|\hSigma_{xx}^{1/2}\bbeta^{k+1} + \frac{1}{\eta}\bgamma^k\|, 1\right\}};\]

\item update for $\bgamma$:
\vspace{-0.5em}
\[\bgamma^{k+1} \leftarrow \bgamma^k + \eta(\hSigma_{xx}^{1/2}\bbeta^{k+1} - \theta^{k+1});\]

\item $k \leftarrow k+1$;
\end{enumerate}
\vspace{-1em}
until $\max\left\{\|\bbeta^{k+1}-\bbeta^k\|, \|\btheta^{k+1}-\btheta^k\|\right\} < \epsilon$.
\end{enumerate}
\vspace{-0.5em}
\noindent{{\bf return} $\hbeta = \bbeta^k$.}
\end{algorithmic}
\label{Algo}
\end{algorithm}

\newpage
\section{Performance of CENet}
\label{Sec:Main}

In this section, we study the performance of CENet.  We first give our main results in Section~\ref{Sec:RateofConv}, where we establish the rate of convergence of CENet under some regularity conditions.  In Section~\ref{Sec:Simulation}, we explore the finite-sample performance of CENet through simulation studies.  Finally, we demonstrate an application of CENet on a real dataset in Section~\ref{Sec:Data}.

\subsection{Rate of Convergence}
\label{Sec:RateofConv}
We first state the main assumptions used to establish the rate of convergence of CENet.

\begin{Assumption}
\label{Assump:PairEllip}
For $j = 1, \ldots, p$, $(x_j, h^*(y))$ is pairwise elliptical.
\end{Assumption}
Under model~\eqref{Eq:ModelTrans}, Assumption~\ref{Assump:PairEllip} holds, for instance, when $(\x, \epsilon)$ follows a joint elliptical distribution.  Such an assumption is sufficient for \eqref{Eq:rho} to hold \citep{Lindskog2003}.  We examine the sensitivity of CENet to violations of Assumption~\ref{Assump:PairEllip} in Section~\ref{Sec:Simulation}. 

\begin{Remark}
More generally, \cite{Kendall1949} showed that when a bivariate random variable $(x, y)$ is not normal but the Pearson correlation $\rho$ exists, an approximation of the Kendall's tau using the bivariate Gram-Charlier series expansion up to the fourth-order cumulants yield
\[\tau \approx \frac{2}{\pi}\sin^{-1}(\rho) + \frac{1}{24\pi(1-\rho^2)^{3/2}}\left\{(\kappa_{40}+\kappa_{04})(3\rho-2\rho^3) - 4(\kappa_{31}+\kappa_{13}) + 6\rho\kappa_{22}\right\},\]
where $\kappa_{40} = \mu_{40}-3$, $\kappa_{31} = \mu_{31} - 3\rho$, $\kappa_{22} = \mu_{22}-2\rho^2-1$ are bivariate cumulants.  This means that we may expect \eqref{Eq:rho} to be approximately true even in some cases where Assumption~\ref{Assump:PairEllip} is violated. 
\end{Remark}


The following subgaussian assumption on $\x$ and boundedness condition on $\bSigma_{xx}$ are adopted to ensure that a restricted eigenvalue condition holds with high probability for $\hSigma_{xx}$.  The restricted eigenvalue condition is a key assumption used in bounding the estimation error of lasso-type estimators in the high-dimensional sparse linear regression model \citep{Bickel2009, Buhlmann2011}.
\begin{Assumption}
\label{Assump:Main}
\emph{}
\begin{enumerate}[(i)]
\item $\x$ is subgaussian, with subgaussian norm $\|\x\|_{\psi_2} := \sup_{\|\bu\| = 1} \|\bu^T\x\|_{\psi_2} \leq K$, where $\|z\|_{\psi_2} = \sup_{p\geq 1}p^{-1/2}(E|z|^p)^{1/p}$;
\item There exists a positive constant $M>1$ such that $M^{-1} \leq \lambda_{\min}(\bSigma_{xx}) < \lambda_{\max}(\bSigma_{xx}) \leq M$.
\end{enumerate}
\end{Assumption}

\begin{Remark}
Assumptions~\ref{Assump:PairEllip} and~\ref{Assump:Main}(i) are satisfied when $(\x, \epsilon)$ is multivariate normal.  More generally, suppose that $(\x, \epsilon)$ follows a mixture variance normal distribution, i.e., $(\x, \epsilon)^T \stackrel{d}{=} \bmu + r\z$, where $\bmu\in\R^{p+1}$ is a constant vector, $\z \sim \cN_{p+1}(\0, \bSigma)$ and $r\geq 0$ is a scalar random variable independent of $\z$.  Suppose further that $0<r<B$ almost surely, where $B>0$ is some constant.  Then $(\x, \epsilon)$ is joint elliptical and $\x$ is subgaussian, and so Assumptions~\ref{Assump:PairEllip} and~\ref{Assump:Main}(i) are also satisfied.
\end{Remark}

We now state our main results:
\begin{Theorem}
\label{Thm:Main}
Consider the linear transformation model~\eqref{Eq:ModelTrans} with true slope vector $\bbeta^*$ satisfying $0 < \|\bbeta^*\|_0\leq s$.  Suppose that Assumptions~\ref{Assump:PairEllip} and \ref{Assump:Main} hold, and 
\begin{equation}
\label{Thm:DimCond}
n \geq C_1s\log p,
\end{equation}
for some sufficiently large constant $C_1>0$.  For any constant $C'>0$, there exists positive constants $\gamma_1, \gamma_2$, and $C$ depending only on $C_1, C', K$, and $M$, such that when the tuning parameters $\alpha_1 = \gamma\sqrt{\log(p)/n}$ for $\gamma\in [\gamma_1, \gamma_2]$ and $\alpha_2\|\bbeta^*\|_\infty \leq \alpha_1$,
\[
\|\hbeta-\bbeta^*\|^2 \leq Cs\alpha_1^2,
\]
with probability at least $1 - 6\exp(-C'\log p)$.
\end{Theorem}

Since the linear transformation model contains the linear regression model as a special case, the rate in Theorem~\ref{Thm:Main} is optimal.  The following is an immediate consequence of Theorem~\ref{Thm:Main}:
\begin{Corollary}
\label{Cor:Main}
Under the same conditions as in Theorem~\ref{Thm:Main}, for any constant $C'>0$, there exist positive constants $\gamma_1, \gamma_2$, and $C$ depending only on $C_1, C', K$, and $M$, such that when $\alpha_1 = \gamma\sqrt{\log(p)/n}$ for $\gamma\in [\gamma_1, \gamma_2]$ and $\alpha_2\|\bbeta^*\|_\infty \leq \alpha_1$,
\[
\left\|\frac{\hbeta}{\|\hbeta\|}-\frac{\bbeta^*}{\|\bbeta^*\|}\right\|^2 \leq Cs\alpha_1^2,
\]
with probability at least $1 - 6\exp(-C'\log p)$.
\end{Corollary}

\subsection{Simulation Studies}
\label{Sec:Simulation}

We study the finite-sample performance of CENet in this section.  Due to the use of Kendall's tau, we expect CENet to be robust to both nonlinear monotone transformations on the response variable and heavy-tailed noise distributions.  In addition, we are interested in the sensitivity of CENet to violations of Assumption~\ref{Assump:PairEllip} in Section~\ref{Sec:RateofConv}.  We examine these issues in the following simulation studies.  For comparison, we also evaluate the performance of the lasso and the smoothed MRC estimator of \cite{HanWang2015} in our simulation studies.  More extensive simulation results can be found in Appendix~D.

\paragraph{\textbf{Simulation settings:}}  We consider two scenarios for the transformation function:\\
\[\emph{Scenario 1}: h^*(y) = y, \qquad\emph{Scenario 2}: h^*(y) = y^{1/3}.\]
The first scenario concerns the linear model whereas the second scenario concerns a linear transformation model.  For each scenario, we generate data according to model~\eqref{Eq:ModelTrans} with $\bbeta^* = (1, 2, 3, 4, 0, \ldots, 0)$, $\x\sim\cN(\0, \bSigma)$ with $\bSigma_{ij} = 0.3^{|i-j|}$, and we sample $\epsilon$ independent of $\x$ from one of the following distributions:
\begin{enumerate}[$\quad$(a)] 
\item normal: $\cN(0, \sigma^2)$;
\item contaminated normal: $0.8\cN(0, \sigma^2) + 0.2\text{Cauchy}(\text{location}=0$, $\text{scale}=10\sigma)$;
\item mixture normal: $0.5\cN(-3, \sigma^2) + 0.5\cN(3, \sigma^2)$;
\item centralized gamma: $\epsilon = \sigma(z - \sqrt{3})$, where $z\sim \Gamma(\text{shape}=3$, $\text{rate}=\sqrt{3})$.
\end{enumerate}
Each of the (uncontaminated) noise distribution above has mean 0 and variance~$\sigma^2$.  We calibrate $\sigma^2$ to achieve a desired level of $R^2 := \Var(\x^T\bbeta^*)/[\Var(\x^T\bbeta^*)+\sigma^2]$.  Under such a data generating mechanism, Assumption~\ref{Assump:PairEllip} is not satisfied except for noise distribution (a).  

Figure~\ref{Figure:ErrorDist} shows the densities for each of the four noise distributions.  Compared to the normal distribution which is unimodal and symmetric, the contaminated normal distribution has heavier tails whereas the mixture normal distribution is bimodal and the gamma distribution is skewed. 
\begin{figure}[htp]
\centering
\includegraphics[angle=0,width=5.5in]{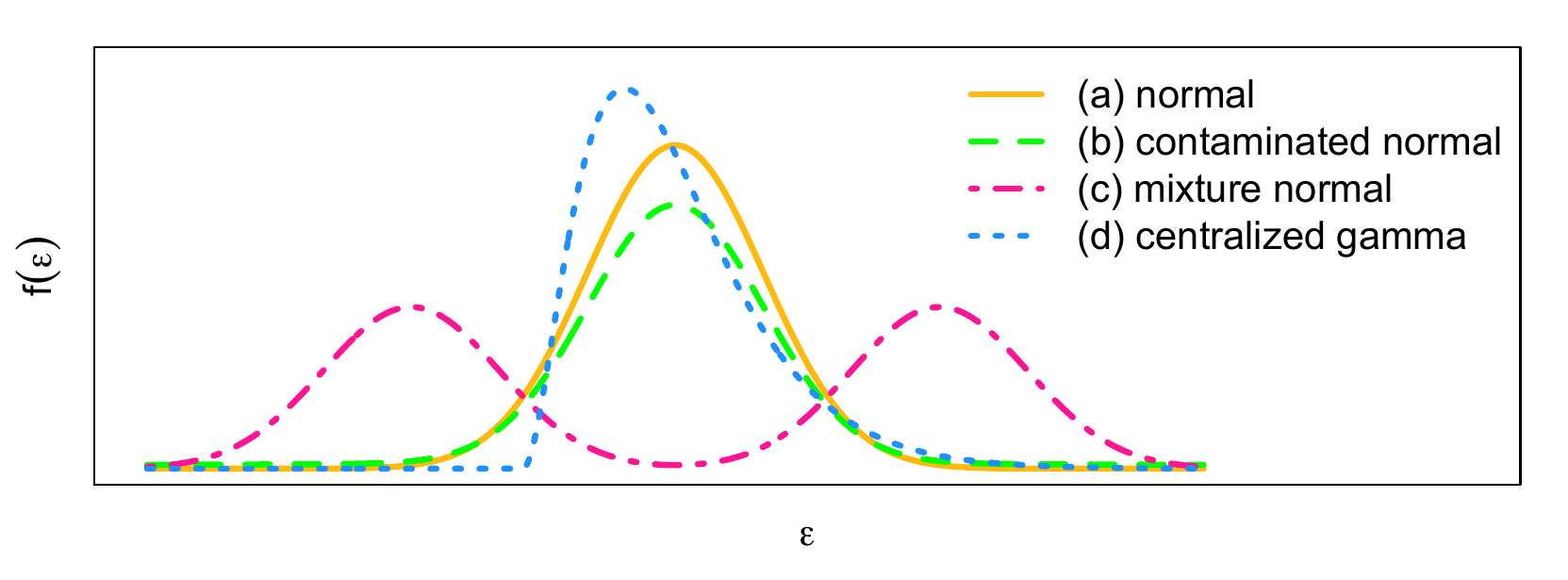}
\caption{Densities for the noise distributions considered in the simulation settings.}
\label{Figure:ErrorDist}
\end{figure}

\vspace{-1em}
\paragraph{\textbf{Performance measures:}} For fair comparison of different methods, we assess the performance of an estimator $\hbeta = (\hat{\beta}_1, \ldots, \hat{\beta}_p)$ by the standardized squared estimation error 
\[
\min \left\{\left\|\frac{\hbeta}{\|\hbeta\|} - \frac{\bbeta^*}{\|\bbeta^*\|}\right\|^2, \left\|\frac{\hbeta}{\|\hbeta\|} + \frac{\bbeta^*}{\|\bbeta^*\|}\right\|^2\right\}.
\] 
In the case $\hbeta = \0$, we simply set the estimation error as $\|\bbeta^*/\|\bbeta^*\|\|^2 = 1$.  Each method we consider has a tuning parameter that controls the sparsity of $\hbeta$ ($\alpha_1$ for CENet, $\lambda$ for lasso and MRC).  We compute the estimation error as a function of the number of non-zero coordinates of $\hbeta$.

To measure the accuracy of support recovery for $\bbeta^* = (\beta_1^*, \ldots, \beta_p^*)$, we also consider plotting the receiver operating characteristic (ROC) curves, which is the plot of the true positive rate (TPR) against the false positive rate (FPR):
\[
\text{TPR} = \frac{|\{j: \hat{\beta}_j \neq 0, \beta^*_j \neq 0\}|}{|\{j: \beta^*_j \neq 0\}|} \quad \text{and} \quad
\text{FPR} = \frac{|\{j: \hat{\beta}_j\neq 0, \beta^*_j = 0\}|}{|\{j: \beta^*_j = 0\}|}. 
\]
TPR gives the proportion of nonzero elements in $\bbeta^*$ that are correctly estimated to be nonzero, whereas FPR gives the proportion of zero elements in $\bbeta^*$ that are incorrectly estimated to be nonzero. 

\paragraph{\textbf{Implementation details:}}
For each scenario and noise distribution, we generate 100 simulated datasets with dimension $p = 100$ and sample size $n = 50$ (high-dimensional setting), $n = 200$ (low-dimensional setting).  For CENet, we fix $\alpha_2 = 10^{-8}, 10^{-7}$, or $10^{-6}$ and consider a range of $\alpha_1$.  In all numerical results reported in this section, we set the ADMM parameter $\eta = 2$ and the tolerance level $\epsilon = 10^{-6}$ in Algorithm~\ref{Algo}.  We implement the smoothed MRC estimator using the code provided by the authors.  MRC requires the specification of a smoothing function and corresponding smoothing parameter value.  We choose a Gaussian CDF approximation and use its default smoothing parameter value, and consider a range of $\lambda$.  For lasso, we use the implementation given in the \texttt{R}-package \texttt{gam} \citep{Hastie2015} and consider a range of $\lambda$.  

\paragraph{\textbf{Results:}}  Since both CENet and MRC are invariant to monotone transformation of the response variable $y$, we combine the results for scenarios 1 and 2 in Figures~\ref{Figure:Rsq0.6n50} and \ref{Figure:Rsq0.6n200}.  Figure~\ref{Figure:Rsq0.6n50} shows the results for each method under different noise distributions, when $n = 50, p=100$, and $R^2 = 0.6$.  The panels on the left column correspond to the plots of average squared estimation error (MSE) over average number of non-zero coordinates of estimated slopes, while the panels on the right column correspond to the average ROC curves.  When the noise term follows a contaminated normal distribution, CENet performs the best while lasso performs the worst under both scenarios 1 and 2.  Like CENet, MRC is much more robust to contaminated normal noise than lasso.  Note that under both scenarios, the lasso estimated slopes tend to be dense.  Moreover,  the MSE is large and the estimation error curves stay relatively flat across different sparsity levels of the estimated slopes.  It is therefore not surprising that lasso also has poor variable selection performance under the contaminated normal noise distribution.  For other noise distributions, when the estimated slope is sparse, CENet has comparable MSE to lasso under scenario 1, and it has much smaller MSE than lasso (under scenario 2) and MRC (under both scenarios).  Similarly, the variable selection performance of CENet is comparable to lasso under scenario 1, and is consistently better than lasso under scenario 2 and MRC under both scenarios.  For all noise distributions, we also see that the overall estimation and variable selection performance of CENet improves as $\alpha_2$ gets smaller.  Based on our limited simulation experience, we recommend $\alpha_2 = 10^{-8}$ when each predictor variable is standardized to have unit variance.  

Although MRC is designed for the high-dimensional setting, we see from Figure~\ref{Figure:Rsq0.6n50} that its MSE curves are quite flat and take relatively large values compared to CENet when $n = 50, p =100$, across all noise distributions.  A plot of its MSE against tuning parameter~$\lambda$ (see Figure~\ref{Figure:Rsq0.6MSE} in Appendix~D) suggests that MRC suffers from convergence issue in such a low signal-to-noise ratio setting.  As we increase the sample size to $n = 200$, we see from Figure~\ref{Figure:Rsq0.6n200} that the convergence issue dissipates.  The estimation error curves of MRC first decrease then increase with respect to the number of non-zero coordinates of estimated slopes, just like those of all other methods (except for lasso under the contaminated normal noise).  The estimation performance of MRC improves significantly compared to when $n = 50$.  We observe similar improvement when $R^2$ is increased from 0.6 to 0.9 (see Appendix~D).  Under the contaminated normal noise distribution, CENet still performs the best and lasso still performs the worst under both scenarios.  For other noise distributions, CENet and MRC perform comparably to lasso under scenario~1 and outperform lasso under scenario~2 in terms of estimation.  CENet performs much better than MRC in terms of variable selection.  Interestingly, the performance of lasso under scenario~2 is not as bad as one might have anticipated when the noise distribution is not heavy-tailed.  We also observe that when $n = 200, p=100$, the performance of CENet is relatively insensitive to the values of~$\alpha_2$.  

In summary, CENet has comparable performance to lasso under the linear model, but is considerably more robust to heavy-tailed noise distributions in both the linear and the transformation model.  Moreover, it is not too sensitive to the underlying noise distributions when the noise is sampled independently from the predictors.  Its performance improves as $\alpha_2$ becomes smaller.  On the other hand, MRC is reasonably robust to both non-linearity in the response variable and heavy-tailed noise distributions, but it does not perform as well as CENet in all the cases we consider, unless the signal-to-noise ratio is high.  

\begin{figure}[htp]
\centering
\includegraphics[angle=0,width=5in]{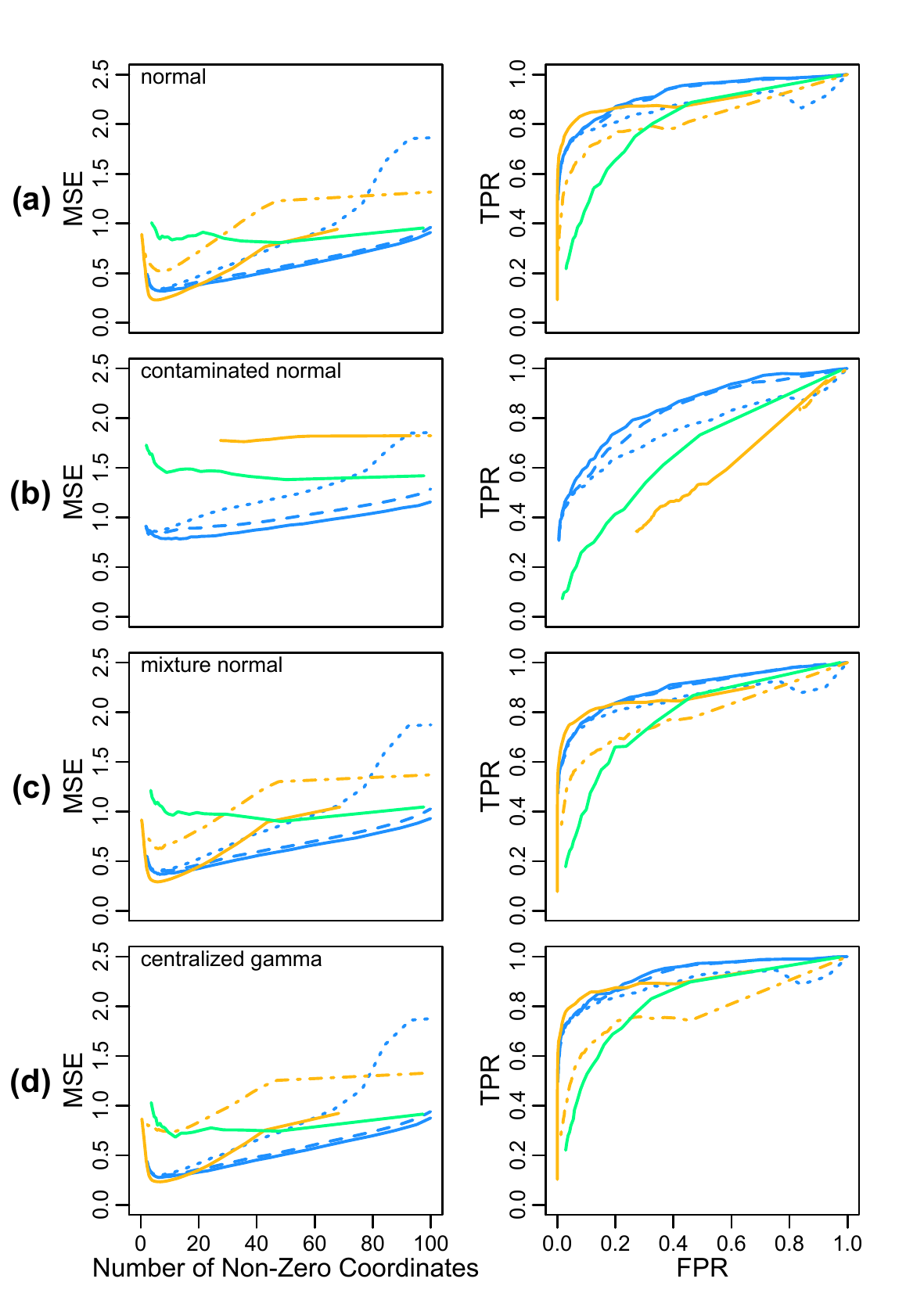}
\caption{Results for lasso under scenarios 1 and 2, CENet under both scenarios, and MRC under both scenarios, averaged over 100 replications, for $n = 50, p = 100, R^2 = 0.6$.  Left column: plot of average standardized squared estimation error against average number of non-zero coordinates of $\hbeta$; right column: plot of average ROC curves.  The rows ordered from top to bottom contain results for noise distributions (a)$-$(d).  Colored lines correspond to lasso under scenario 1 (\protect\includegraphics[height=0.25em]{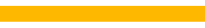}), lasso under scenario 2 (\protect\includegraphics[height=0.25em]{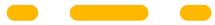}), CENet with $\alpha_2 = 10^{-8}$ (\protect\includegraphics[height=0.25em]{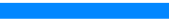}), $\alpha_2 = 10^{-7}$ (\protect\includegraphics[height=0.25em]{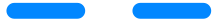}), $\alpha_2 = 10^{-6}$ (\protect\includegraphics[height=0.25em]{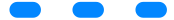}) under both scenarios, and MRC (\protect\includegraphics[height=0.25em]{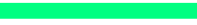}) under both scenarios.}
\label{Figure:Rsq0.6n50}
\end{figure}

\begin{figure}[htp]
\centering
\includegraphics[angle=0,width=5in]{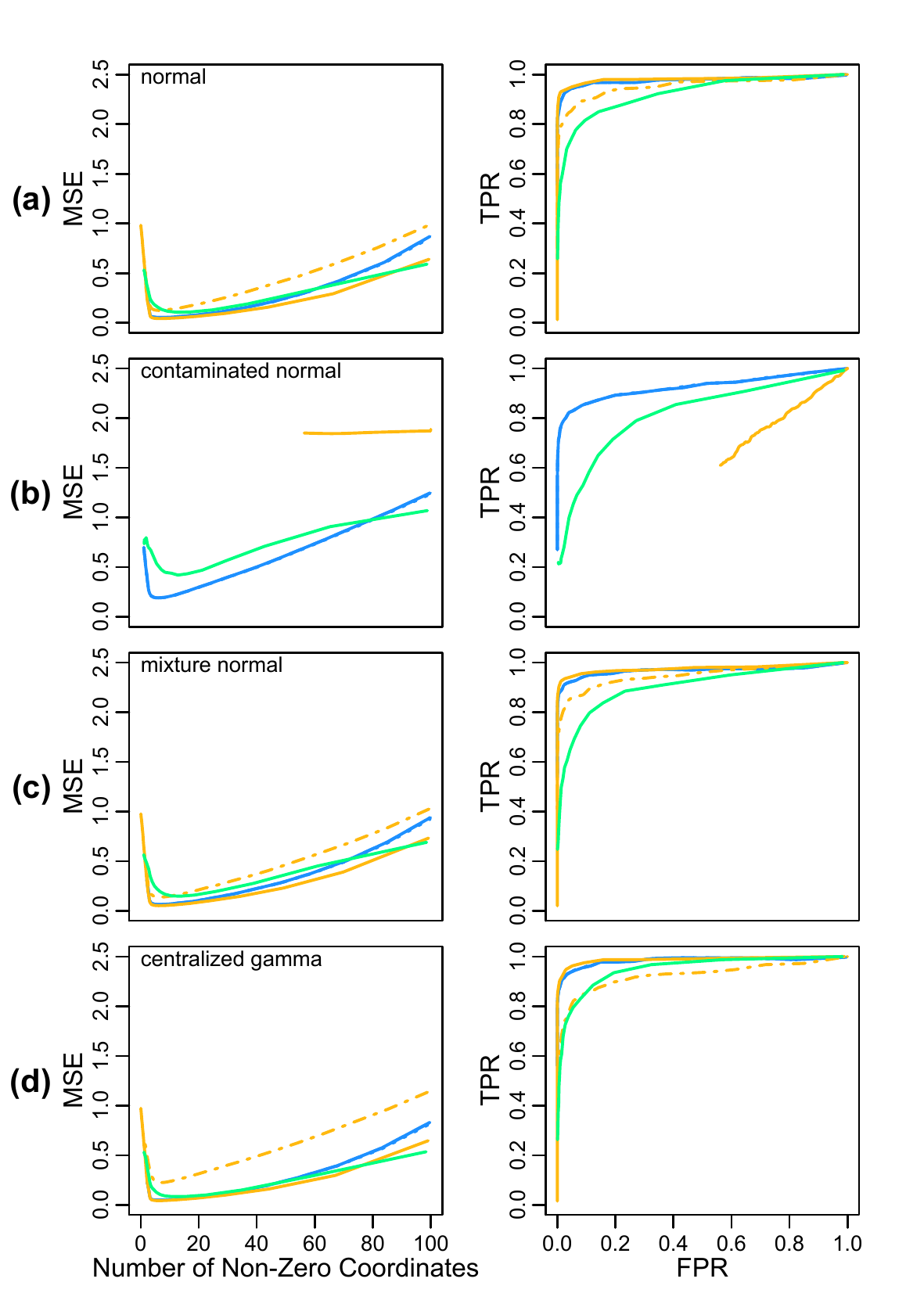}
\caption{Results for lasso under scenarios 1 and 2, CENet under both scenarios, and MRC under both scenarios, averaged over 100 replications, for $n = 200, p = 100, R^2 = 0.6$.  Left column: plot of average standardized squared estimation error against average number of non-zero coordinates of $\hbeta$; right column: plot of average ROC curves.  The rows ordered from top to bottom contain results for noise distributions (a)$-$(d).  Colored lines correspond to lasso under scenario 1 (\protect\includegraphics[height=0.25em]{o1.png}), lasso under scenario 2 (\protect\includegraphics[height=0.25em]{o4.png}), CENet with $\alpha_2 = 10^{-8}$ (\protect\includegraphics[height=0.25em]{b1.png}), $\alpha_2 = 10^{-7}$ (\protect\includegraphics[height=0.25em]{b2.png}), $\alpha_2 = 10^{-6}$ (\protect\includegraphics[height=0.25em]{b3.png}) under both scenarios, and MRC (\protect\includegraphics[height=0.25em]{g1.png}) under both scenarios.}
\label{Figure:Rsq0.6n200}
\end{figure}

\newpage
\subsection{Data Application}
\label{Sec:Data}

We consider the cardiomyopathy microarray data from a transgenic mouse model of dilated cardiomyopathy that is related to the understanding of human heart diseases \citep{Redfern2000}.  The dataset includes an $n\times p$ matrix of gene expression values $\X = (x_{ij})$, where $x_{ij}$ is the expression level of the $j^{th}$ gene for the $i^{th}$ mouse, for $1\leq i\leq n=30$ and $1\leq j \leq p=6319$.  Each mouse also give an outcome expression level $y_i$ of a G protein-coupled receptor, designated Ro1.

The cardiomyopathy microarray data was previously studied in \cite{Segal2003}, \cite{Hall2009}, and \cite{Li2012}.  The goal is to determine genes whose expression changes are due to the expression of Ro1.  However, such a selection is inherently difficult due to two distinguishing features of the microarray data: (i) $p$ vastly exceeds $n$, and (ii) the gene expression values are highly correlated.  \cite{Segal2003} gave an overview and evaluation of some proposals for tackling these issues, and obtained results that motivate regularized regression procedures such as the lasso and the least angle regression \citep{Efron2004} as alternate regression tools for microarray studies.  As noted in \cite{Segal2003}, a major challenge with the application of these high-dimensional regression methods is the determination of the right amount of regularization.  

In this study, we do not argue whether regularized regression procedures are the best tools for analyzing the cardiomyopathy microarray data.  Instead, we are interested in exploring whether there is a significant overlap between the genes selected by CENet and those that are identified in \cite{Segal2003}, \cite{Hall2009}, and \cite{Li2012}.  We address the problem of proper regularization with the stability selection of \cite{Buhlmann2010}.  We first standardize the expression values of each gene to have zero mean and unit variance.  We then repeat for $B=1000$ times:
\begin{enumerate}[(i)]
\item Sample~$n$~pairs of $(\x_i, y_i)$ with replacement from the original dataset.
\item Since $p\gg n$, we apply the robust rank correlation based screening procedure of \cite{Li2012} to the resampled dataset and retain the top 50 genes that have the highest Kendall's tau correlation (in magnitude) with the Ro1 expression values.  \cite{Li2012} showed that for the purpose of variable selection it is appropriate to perform variables screening before fitting a linear transformation model, as the screening step will eliminate irrelevant predictors and keep relevant predictors with high probability.  A similar screening strategy based on sample correlation has been widely used in practice when fitting a linear regression model, and a theoretical justification can be found in \cite{FanLv2008}.  
\item Apply CENet with $\alpha_2 = 10^{-8}$ over a grid of $\alpha_1\in(10^{-6}, 10^{-2})$, and record the genes selected (i.e., genes whose corresponding $\hat{\beta}_j\neq 0$) for each $\alpha_1$. 
\end{enumerate}
Finally, we compute $\hat{p}_j^{\alpha_1}$, the selection frequency of the $j^{th}$ gene across the B resampled datasets at tuning parameter value $\alpha_1$, for $j = 1, \ldots, p$ at each $\alpha_1$ value.   Figure~\ref{Figure:admmh} shows the stability path of each gene, which is the plot of $\hat{p}_j^{\alpha_1}$ against $\alpha_1$ for $j = 1, \ldots, p$.   There are five genes that stand out throughout their stability paths (colored in red): Msa.10108.0, Msa.10044.0, Msa.1009.0, Msa.1166.0, Msa.10274.0.  Besides for Msa.1166.0, none of the other four genes have been previously implicated for the cardiomyopathy dataset.  Msa.1166.0 was identified in \cite{Hall2009} and \cite{Li2012} for its strong nonlinear association with the Ro1 expression level.  A second group of genes stand out when $\alpha_1\approx 0$, and their stability paths are colored in blue: Msa.7019.0, Msa.2877.0, Msa.5794.0, Msa.15442.0.  We highlighted the stability paths of a third group of genes in green, as they also stand out, albeit less strongly, when $\alpha_1\approx 0$: Msa.1590.0, Msa.2134.0, Msa.3969.0.  The genes Msa.7019.0, Msa.2877.0, Msa.2134.0 were identified in \cite{Li2012} for their high Kendall's tau correlations with the Ro1 expression level.    

For comparison, we also compute the stability paths of lasso for the cardiomyopathy dataset.  We follow essentially the same steps as outlined above for CENet, except that we use Pearson correlation rather than Kendall's tau for variables screening, and we apply lasso over $\lambda\in (10^{-6}, 5)$.  From Figure~\ref{Figure:lasso}, we see that four genes stand out throughout their stability paths (colored in red): Msa.2877.0, Msa.964.0, Msa.778.0\_i, Msa.2134.0.  Besides for Msa.964.0, the other three genes have been previously identified by \cite{Segal2003} using lasso with cross-validation.  We also see that four other genes stand out when $\lambda\approx 0$ (colored in blue): Msa.1590.0, Msa.1043.0, Msa.3041.0, Msa.741.0.  Msa.1590.0 is also identified by CENet when $\alpha_1\approx 0$.

Surprisingly, we do not see an overlap between the top five genes identified by CENet and the top four genes identified by lasso, though it is known \citep{Li2012} that Msa.1166.0 identified by CENet is highly collinear with Msa.2877.0 identified by lasso.  Upon inspection, we find that the relative strengths of association between the expression level of each of these nine genes and that of Ro1 change when the measure of association is changed from Pearson correlation to Kendall's tau.  For instance, the two genes Msa.964.0 and Msa.778.0\_i stand out when the association is measured by Pearson correlation, but they do not stand out when the association is measured by Kendall's tau.  Figures~\ref{Figure:scatter1} and~\ref{Figure:scatter2} show the scatterplots of expression level of each of these nine genes with the Ro1 expression level.  In such a case, Pearson correlation measures the strength of linear relationship between the gene expression level and the Ro1 expression level.  For comparison, we also include the scatterplots of the ranked gene expression levels, which serve to illustrate the strengths of association as measured by Kendall's tau.  We see that the association is better captured by Kendall's tau for Msa.10274.0 and Msa.1166.0 (i.e., the scatterplots of ranked expression levels exhibit more linearity than their unranked counterparts), whereas the association is better captured by Pearson correlation for Msa.2134.0 and Msa.2877.0 (i.e., the scatterplots of original expression levels exhibit more linearity than their ranked counterparts).  

Due to the relatively small sample size ($n = 30$), it is hard to justify if there is a single best way to summarize the association between the expression level of Ro1 and that of other genes.  As was mentioned in \cite{Segal2003}, with the current state of microarray technology and study dimensions, different data processing and/or modeling approaches can give very different results.  This is not necessarily a bad thing as such results can be viewed from a ``sensitivity analysis" perspective, and ultimately, microarray results need to be validated experimentally by another complimentary method.  The interpretation of the biological relevance of the genes identified above is beyond the scope of the paper.  For completeness, we include their names and descriptions in Table~\ref{genenames}. 

\begin{figure}[htp]
\centering
\includegraphics[angle=0,width=3.5in]{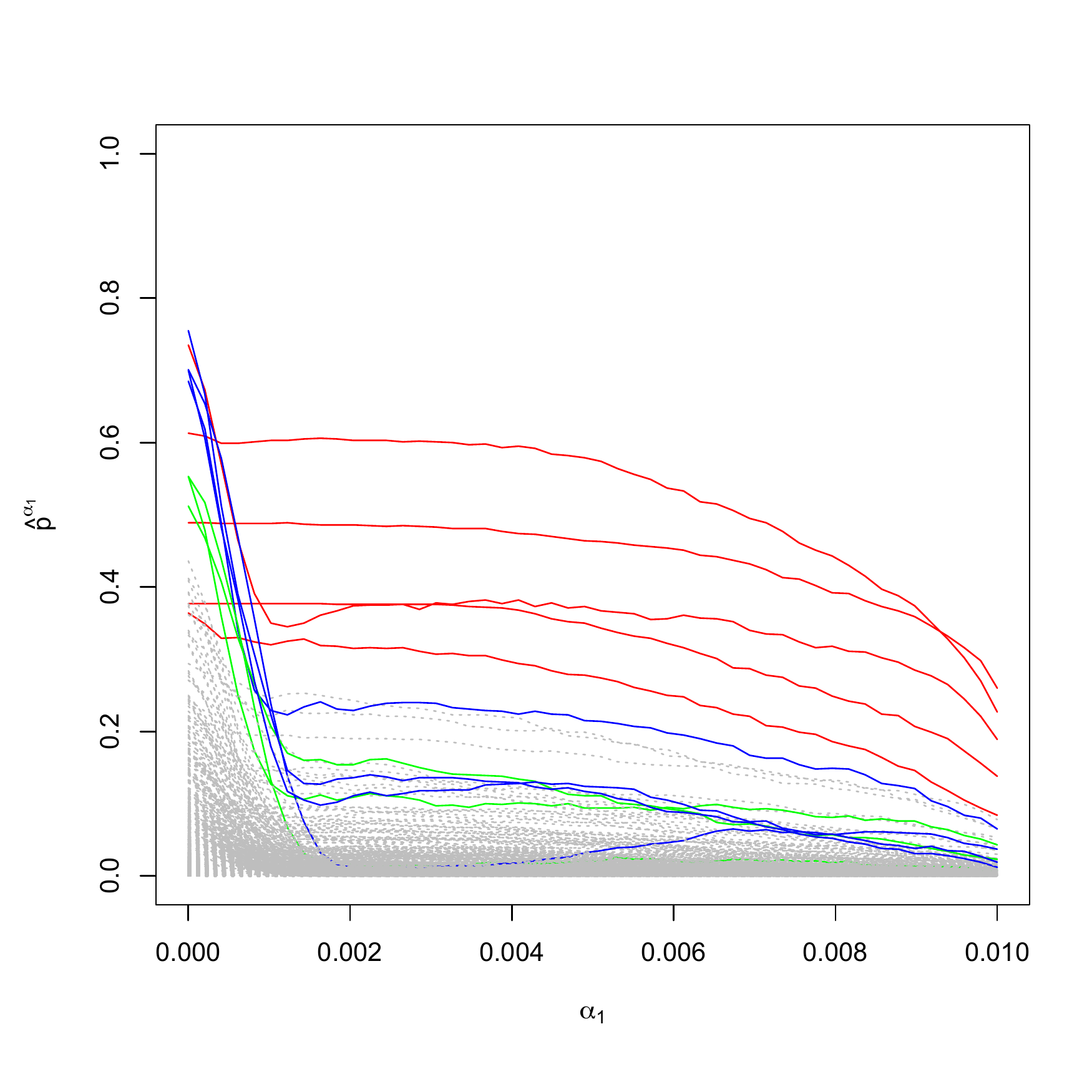}
\vspace{-2em}
\caption{Stability paths of CENet for the cardiomyopathy dataset.  The highlighted stability paths correspond to the genes Msa.10108.0, Msa.10044.0, Msa.1009.0, Msa.1166.0, Msa.10274.0 (\protect\includegraphics[height=0.25em]{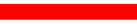}), Msa.7019.0, Msa.2877.0, Msa.5794.0, Msa.15442.0 (\protect\includegraphics[height=0.25em]{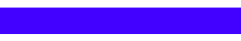}), and Msa.1590.0, Msa.2134.0, Msa.3969.0 (\protect\includegraphics[height=0.25em]{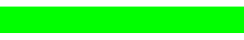}).}
\label{Figure:admmh}
\end{figure}

\begin{figure}
\centering
\includegraphics[angle=0,width=3.5in]{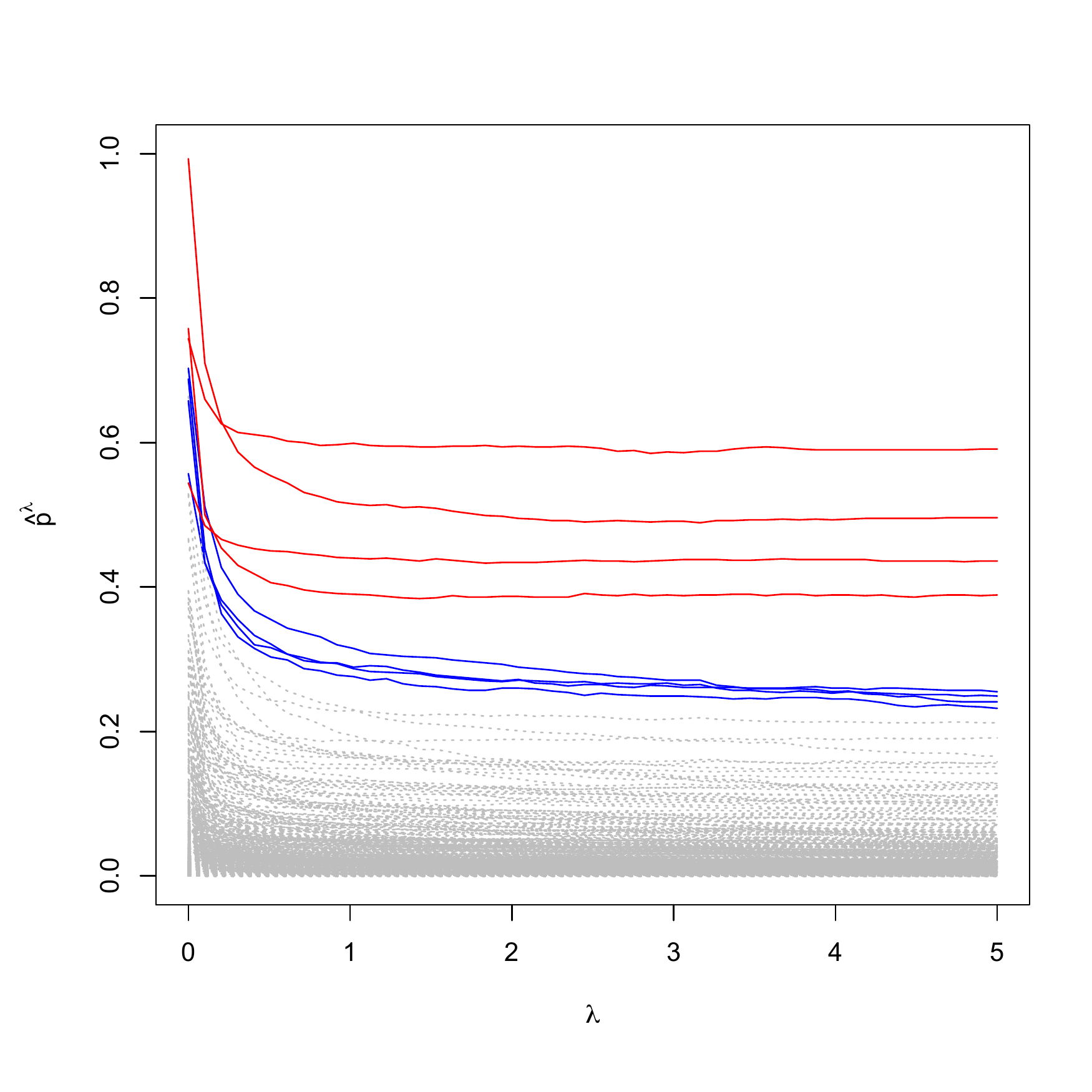}
\vspace{-2em}
\caption{Stability paths of lasso for the cardiomyopathy dataset.  The highlighted stability paths correspond to the genes Msa.2877.0, Msa.964.0, Msa.778.0\_i, Msa.2134.0 (\protect\includegraphics[height=0.25em]{red.png}), and Msa.1590.0, Msa.1043.0, Msa.3041.0, Msa.741.0 (\protect\includegraphics[height=0.25em]{blue2.png}).}
\label{Figure:lasso}
\end{figure}  

\begin{figure}[htp]
\centering
\includegraphics[angle=0,width=5in]{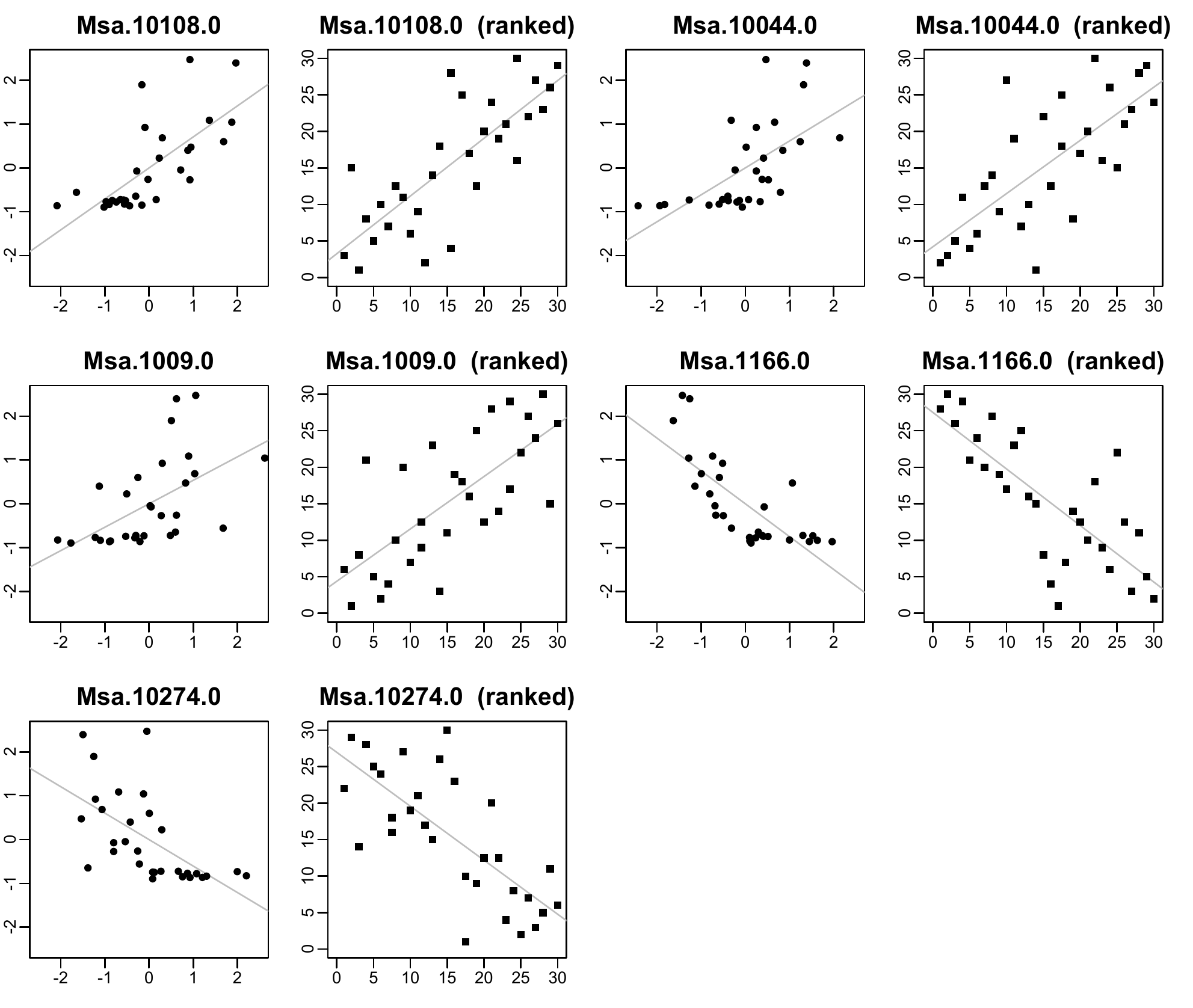}
\caption{Top five genes selected by CENet.  Odd columns: scatterplots of expression level of each gene with that of Ro1; even columns: scatterplots of ranked expression level of each gene with that of Ro1.  Grey lines in each panel indicate the least squares lines.}
\label{Figure:scatter1}
\end{figure}

\begin{figure}[htp]
\centering
\includegraphics[angle=0,width=5in]{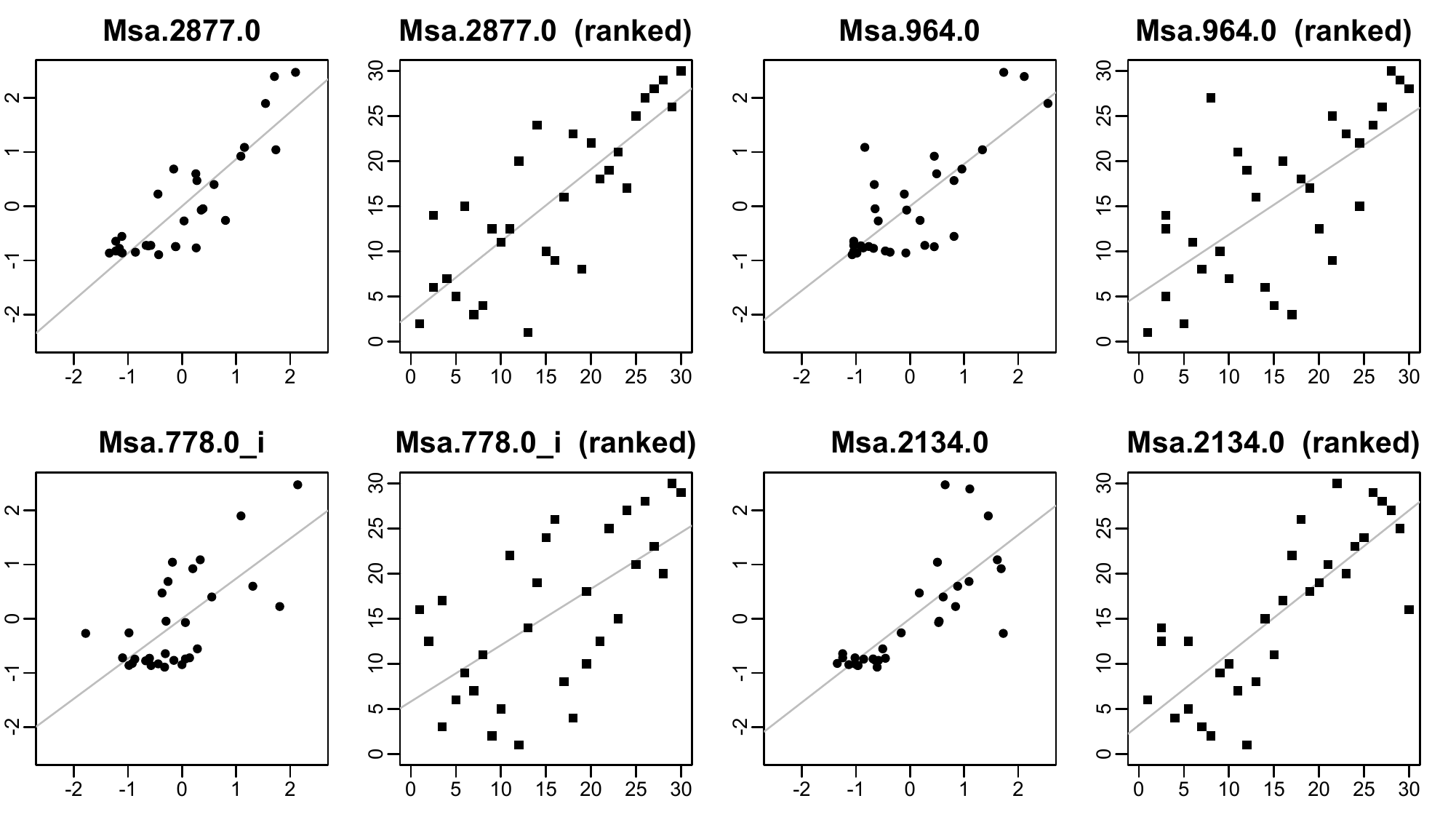}
\caption{Top four genes selected by lasso.  Odd columns: scatterplots of expression level of each gene with that of Ro1; even columns: scatterplots of ranked expression level of each gene with that of Ro1.  Grey lines in each panel indicate the least squares lines.}
\label{Figure:scatter2}
\end{figure}

\begin{table}[htp]
\centering
\begin{tabular}{|l|p{11cm}|}
\hline
Genes & Descriptions  \\ 
\hline
Msa.10108.0 & Homologous to sp P07814: MULTIFUNCTIONAL AMINOACYL-TRNA SYNTHETASE\\  
Msa.10044.0 & Homologous to sp P10658: PROBABLE PHOSPHOSERINE AMINOTRANSFERASE\\
Msa.1009.0 & Mouse myoblast D1 (MyoD1) mRNA, complete cds\\
Msa.1166.0 & Mus musculus sterol carrier protein-2 (SCP-2) gene, complete cds\\
Msa.10274.0 & Homologous to sp P07335: CREATINE KINASE, B CHAIN\\
Msa.7019.0 & Homologous to sp P10719: ATP SYNTHASE BETA CHAIN, MITOCHONDRIAL PRECURSOR\\
Msa.2877.0 & Mouse MARib mRNA for ribophorin, complete cds \\
Msa.5794.0 & Homologous to sp P11730: CALCIUM/CALMODULIN-DEPENDENT PROTEIN KINASE TYPE II GAMMA CHAIN (CAM- KINAS\\
Msa.15442.0 & Homologous to sp P07379: PHOSPHOENOLPYRUVATE CARBOXYKINASE, CYTOSOLIC (GTP)\\
Msa.1590.0 & Mouse protein kinase C delta mRNA, complete cds\\
Msa.2134.0 & Murine mRNA for 4F2 antigen heavy chain\\
Msa.3969.0 & Homologous to sp P36542: ATP SYNTHASE GAMMA CHAIN, MITOCHONDRIAL PRECURSOR\\
\hline
\end{tabular}
\caption{Names and descriptions of genes identified in the CENet stability paths.}
\label{genenames}
\end{table}     

\newpage
\section{Discussion}
\label{Sec:Discussion}

Nonlinear transformation of variables is commonly practiced in regression analysis.  In the mainstream statistical literature, variables transformation methods have been introduced by way of an estimation problem based on a model with parameters and noise terms \citep{BoxCox1964, Carroll1984}.  The CENet method proposed in this paper and many other methods referenced in Section~\ref{Sec:Intro} fall into this category.  Another line of work focuses on dimension reduction rather than approximation of regression surfaces, and an early important work is the sliced inverse regression (SIR) method of \cite{DuanLi1991}.  Yet a different approach towards nonlinear regression analysis exists, where one searches for optimal nonlinear transformation of variables through optimizing some objective function of interest, with or without a model in mind.  This includes, in particular, the alternating conditional expectation (ACE) algorithm proposed by \cite{Breiman1985}.  Although the ACE algorithm is a powerful and useful tool, it is more of a nonlinear multivariate analysis method than a regression method \citep{Buja1990}, and is known to be connected to the additive principal components (APC) method of \cite{DonnellBuja1994}.  Despite the fact that CENet is developed with model~\eqref{Eq:ModelTrans} in mind, interestingly we are able to find a connection between CENet and SIR, ACE, APC.  Such a connection stems primarily from the fact that a correlation-based approach is used for all these methods, and that we take a random-$\x$ view when motivating the CENet methodology.

\subsection{Connections between CENet and SIR, ACE, APC} 
\label{SubSec:SIR} 

The inquiry of connections between CENet and SIR, ACE, APC is motivated by the key observation in Proposition~\ref{Prop:OurRep} that is used to formulate CENet.  Setting aside the fact that relation~\eqref{Eq:rho} (whose validity requires distributional assumption on $(\x, \epsilon)$) is used to bypass the unknown nonlinear transformation $h^*$ in model~\eqref{Eq:ModelTrans}, on a population level CENet is built on the simple idea of performing CCA on $\x$ and $h^*(y)$, allowing linear transformation on each.  Such an observation naturally links CENet to SIR, ACE, and APC, three methods that we will review in the next few paragraphs.

In their pioneering article on dimension reduction, \cite{DuanLi1991} considered a semiparametric index model of the form
\begin{equation}
\label{Eq:SIRmodel}
y = f(\x^T\bbeta_1^*, \ldots, \x^T\bbeta_k^*, \epsilon),
\end{equation}
where $f$ is an unknown link function and $\epsilon$ is a noise term independent of $\x$, and proposed SIR for estimating the effective dimension reduction (e.d.r.) space, which is the linear space generated by the $\bbeta_i^*$'s.  Rather than modeling the conditional distribution of $y$ given $\x$ as in conventional approaches, \cite{DuanLi1991} took an inverse modeling perspective by treating $\x$ as random and considering the conditional distribution of $\x$ given $y$.  Indeed, when $k=1$, the population version of SIR amounts to solving the following generalized eigenproblem:
\begin{equation}
\label{Eq:SIR}
\max_{\bbeta\in\R^p}\frac{\bbeta^T\Cov[E(\x|y)]\bbeta}{\bbeta^T\bSigma_{xx}\bbeta}.
\end{equation}
\cite{ChenLi1998} later showed that there is a natural connection between SIR and multiple linear regression, in that the SIR direction obtained from \eqref{Eq:SIR} can be interpreted as the slope vector of multiple linear regression applied to an optimally transformed $y$.  That is, the solution to \eqref{Eq:SIR} also solves
\begin{equation}
\label{Eq:SIRCCA}
\max_{h\in L^2(y), \bbeta\in\R^p} \Corr(h(y), \x^T\bbeta),
\end{equation}
where the maximization is taken over all transformations $h$ with $E(h^2(y))<\infty$ and vectors $\bbeta\in\R^p$.  Regarding the statistical theory of SIR, \cite{DuanLi1991} showed that the SIR directions fall into the e.d.r. space under a linearity assumption on the distribution of $\x$, which is satisfied when $\x$ follows an elliptical distribution.  Observing that model~\eqref{Eq:ModelTrans} is a special case of model \eqref{Eq:SIRmodel} for which SIR is applicable, that both SIR and CENet require some sort of (elliptical) distributional assumption on $(\x, \epsilon)$ for recovering $\bbeta^*$, and also comparing \eqref{Eq:Corrstar} and \eqref{Eq:SIRCCA}, it is natural to expect that a connection exists between SIR and CENet. 

The characterization \eqref{Eq:SIRCCA} of SIR in terms of maximum correlation connects it to ACE \citep{Breiman1985} and APC \citep{DonnellBuja1994, Tan2015}, two nonlinear multivariate techniques that make use of maximum correlation for searching optimal variable transformations.  Given a response $y$ and predictors $x_1, \ldots, x_p$, \cite{Breiman1985} considered a regression model of the form $h^*(y) = \sum_{j=1}^p g_j^*(x_j) + \epsilon$,
and defined the optimal ACE transformations as functions $h^{**}, g_1^{**}, \ldots, g_p^{**}$ that solve
\[
\min_{h\in L^2(y), g_1\in L^2(x_1), \ldots, g_p\in L^2(x_p)} \frac{\Var(h(y) - \sum_{j=1}^pg_j(x_j))}{\Var(h(y))}. 
\]
Whereas each predictor is allowed to make only linear transformation in SIR, each is allowed to make possibly nonlinear transformation in ACE.  When $g_j^*(x_j), j = 1, \ldots, p$ have a joint normal distribution and $\epsilon$ is an independent $\cN(0, \sigma_\epsilon^2)$, $h^{**} = h^*$ and $g_j^{**} = g_j^*$, for $j = 1, \ldots, p$ \citep[p. 581]{Breiman1985}.  

On the other hand, APC treats all variables symmetrically by fitting an implicit additive equation $\sum_{j=1}^p g_j(x_j) = 0$,
as opposed to regression analysis which treats variables asymmetrically by singling out a response variable.  The (smallest) APC transformations can be found by solving
\[
\min_{g_1\in L^2(x_1), \ldots, g_p\in L^2(x_p)} \frac{\Var(\sum_{j=1}^pg_j(x_j))}{\sum_{j=1}^p\Var(g_j(x_j))}.
\]
In the case where $p=2$, the smallest and the largest APCs (obtained by maximizing instead of minimizing the criterion above) are in trivial correspondence since a largest APC given by $(g_1^{**}, g_2^{**})$ generates a smallest APC given by $(g_1^{**}, -g_2^{**})$.  

In general, APC and ACE analyses are not identical, but a direct correspondence exists in one-simple situation: single-predictor ACE is equivalent to two-variable APC \citep[Section 4.6]{DonnellBuja1994}.  Such an equivalence enables us to link SIR to APC, by treating $\x$ and $y$ as the two variables of interest, and considering linear transformation for $\x$ and possibly nonlinear transformation for $y$.  The following lemma summarizes the connection among SIR, ACE, and APC.

\begin{Proposition}
\label{Prop:ConnectACEAPC}
SIR, ACE, and APC are equivalent in that
\begin{align}
\argmax_{h\in L^2(y), \bbeta\in\R^p} \Corr(h(y), \x^T\bbeta) 
&= \argmin_{h\in L^2(y), \bbeta\in\R^p} \frac{\Var(h(y) - \x^T\bbeta)}{\Var(h(y))} \nonumber\\
&= \argmax_{h\in L^2(y), \bbeta\in\R^p} \frac{\Var(h(y) + \x^T\bbeta)}{\Var(h(y)) + \Var(\x^T\bbeta)}. \label{Eq:ACEAPC}
\end{align}
\end{Proposition}

\begin{Remark}
In \eqref{Eq:ACEAPC} we replace the $p$ infinite-dimensional Hilbert spaces $L^2(x_1), \ldots, L^2(x_p)$ in ACE by a single $p$-dimensional Hilbert space $\Hsp_\x = \{\x^T\bbeta: \bbeta\in\R^p\}$ and keep $L^2(y) = \Hsp_y := \{h(y): \Var(h(y))<\infty\}$ the same.  On the other hand, in APC we take $p=2$ for the two variables $\x$ and $y$ and consider optimizing over $g_1(\x)\in\Hsp_\x$ and $g_2(y)\in\Hsp_y$.  In essence, one can view SIR as performing CCA between the Hilbert spaces $\Hsp_\x$ and $\Hsp_y$, and all three criteria in \eqref{Eq:ACEAPC} are equivalent characterization of the CCA problem.  
\end{Remark} 

Now that we have established the equivalence of SIR, ACE, and APC (without referring to any true underlying models), we are ready to explore their connection with CENet.  For brevity, in the rest of the section we shall refer to the solution $(h, \bbeta)$ to \eqref{Eq:ACEAPC} as the SIR transformation and the SIR direction, respectively.  Comparing \eqref{Eq:SIRCCA} with \eqref{Eq:Corrstar}, we see that to draw a connection between SIR and CENet, it suffices to find conditions for which the SIR transformation is $h^*$ (up to intercept and scale) when model~\eqref{Eq:ModelTrans} holds.  The following proposition gives two such conditions:

\begin{Proposition}
\label{Prop:ConnectSIR}
Under model~\eqref{Eq:ModelTrans}, suppose that $h^*$ is invertible, $\epsilon$ has mean 0 and variance $\sigma_\epsilon^2$, and that one of the following holds:
\begin{enumerate}[(i)] \itemsep -0.3em
\item $(\x, \epsilon)$ follows a joint elliptical distribution with $\Cov(\x, \epsilon) = \0$;
\item $\epsilon$ is degenerate, i.e., $\sigma_\epsilon^2 = 0$ (no distributional assumption on~$\x$).
\end{enumerate}
Then the solutions to \eqref{Eq:ACEAPC} recover $h^*$ and $\bbeta^*$ in model~\eqref{Eq:ModelTrans}.  In other words, the SIR direction $\bbeta$ obtained from \eqref{Eq:ACEAPC} also solves \eqref{Eq:Corrstar}. 
\end{Proposition}

Condition (i) includes the special case where $\x$ has a multivariate normal distribution and $\epsilon\sim\cN(0, \sigma_\epsilon^2)$ independent of $\x$.  Under condition (i), both the CENet direction and the SIR direction are proportional to $\bbeta^*$.  However, the validity of CENet relies on the relationship between the orthant probability $P(x_j>0, \x^T\bbeta^*+\epsilon>0)$ and $\rho_j$, whereas the validity of SIR relies on the linearity of the conditional expectation $E(\x|\x^T\bbeta^*+\epsilon)$.  Interestingly, when $h^*$ is invertible, $\x$ has a joint elliptical distribution, and $\epsilon$ is independent of $\x$, then the SIR direction is always proportional to $\bbeta^*$ regardless of the distribution of $\epsilon$ \citep{DuanLi1991}.  However, in such a case the SIR transformation is not necessarily given by $h^*$.  On the other hand, the CENet direction will be affected by the distribution of $\epsilon$ even if $\epsilon$ is independent of $\x$, but the relation \eqref{Eq:rho} (and hence, CENet) is valid under condition (i) even when $\Cov(\x, \epsilon) \neq \0$.  Although SIR is applicable to the more general model \eqref{Eq:SIRmodel} and requires less restrictive assumption on the noise distribution than CENet, as mentioned in the introduction its rate of convergence in the high dimensions has not been derived, whereas CENet is shown to attain the optimal rate of convergence under a pair-elliptical assumption on $(x_j, h^*(y))$, $j = 1, \ldots, p$.


Under condition (ii), SIR recovers $\bbeta^*$ in model~\eqref{Eq:ModelTrans} regardless of the distribution of $\x$.  This is not surprising given the equivalence of SIR with \eqref{Eq:SIRCCA}.  Unfortunately, for CENet, an additional assumption like ellipticality of the distribution of $\x$ is required for recovering $\bbeta^*$ even if $\epsilon$ is degenerate.  This is because such an assumption is critical for us to bypass the estimation of $h^*$ in model~\eqref{Eq:ModelTrans}, and the convenience comes at a price.   On the contrary, the explicit need to estimate $h^*$ in \eqref{Eq:ACEAPC} grants greater flexibility on the distribution of $\x$ for SIR/ACE/APC, but it also renders their implementations more cumbersome and their theoretical analyses more challenging.  Even if one considers computing the SIR direction from \eqref{Eq:SIR} where estimation of $h^*$ is not needed, the slicing step \citep{DuanLi1991} involved in estimating $\Cov[E(\x|y)]$ can still affect the finite-sample performance of SIR when $p$ is large relative to $n$ (even if $p<n$) due to data sparseness.

To summarize, we have the following corollary:
\begin{Corollary}
\label{Cor:Connect}
Under model~\eqref{Eq:ModelTrans}, suppose that $h^*$ is invertible, $\epsilon$ has mean 0 and variance $\sigma_\epsilon^2$, and condition (i) in Proposition \ref{Prop:ConnectSIR} holds.
Then SIR, ACE, APC, and CENet are equivalent. 
\end{Corollary}

\begin{Remark}
Under conditions given in Corollary~\ref{Cor:Connect}, unlike SIR, ACE, and APC, CENet recovers $\bbeta^*$ without requiring knowledge of $h^*$.
\end{Remark}

\subsection{Concluding Remarks}
We have presented a new approach, CENet, for efficiently performing simultaneous variables selection and slope vector estimation in the linear transformation model when the dimension $p$ can be much larger than the sample size $n$.  CENet is the solution of a convex optimization problem, for which a computationally efficient algorithm with guaranteed convergence to the global optimum is provided.  CENet has two tuning parameters that control the level of regularization.  We showed that under certain regularity conditions, CENet attains the optimal rate of convergence for the estimation of slope vector in the sparse linear transformational models, just as the best regression method do in the sparse linear models.  On the other hand, CENet does not require the specification of the transformation function and the noise distribution, and so is considerably more robust compared to existing estimation and variable selection methods for linear models in the large-$p$-small-$n$ setting.  The robustness of CENet to both nonlinearity in the response and heavy-tailedness in the noise distribution is confirmed by the theoretical results and simulation studies given in the paper.  We also showed that a connection exists between CENet and three other nonlinear regression/multivariate methods: SIR, ACE, APC, when model \eqref{Eq:ModelTrans} is valid and an elliptical assumption on $(\x, \epsilon)$ holds.  In such a case, an advantage of CENet over SIR, ACE, and APC is that it bypasses the need to explicitly estimate the unknown transformation function when the primary interest is in the estimation of the slope vector.

The rate of convergence of the CENet estimator in the paper is established under an elliptical assumption on the predictor-transformed-response pairs.  As demonstrated in our simulation results, CENet performs well when the noise term is generated independently of the predictor but not necessarily from the same family of distribution.  This motivates studying the convergence rate of CENet when the pair-elliptical assumption is violated, and exploring if there exists more general condition for which CENet still attains the optimal rate of convergence.  We leave this for future work.

\section*{Acknowledgements}
We thank Fang Han and Honglang Wang for providing R code for implementing the MRC estimator proposed in \cite{HanWang2015}; Mark Segal for sharing the cardiomyopathy dataset studied in \cite{Segal2003}; Andreas Buja, Zijian Guo, Po-Ling Loh, and Zongming Ma for helpful feedback that has helped improve the presentation of the paper.

\bibliographystyle{agsm}
\bibliography{Ref.bib}

\end{document}